\documentclass[12pt]{article}
\title{\LARGE \bf Stability of Disturbance Based Unified Control}
\author{Oleg O. Khamisov$^{1}$,
\thanks{$^{1}$Oleg O. Khamisov is with Center for Energy Systems,
        Skolkovo Institute of Science and Technology, Skolkovo Innovation Center, Building 3, Moscow  143026, Russia
        {\tt\small oleg.khamisov@skolkovotech.ru}}
        Janusz Bialek$^{2}$ and
\thanks{$^{2}$Janusz Bialek is with Center for Energy Systems,
        Skolkovo Institute of Science and Technology, Skolkovo Innovation Center, Building 3, Moscow  143026, Russia
        {\tt\small J.Bialek@skoltech.ru}}
        Anatoly Dymarsky$^{3}$
\thanks{$^{3}$Anatoly Dymarsky is with Center for Energy Systems,
        Skolkovo Institute of Science and Technology, Skolkovo Innovation Center, Building 3, Moscow  143026, Russia
        {\tt\small A.Dymarsky@skoltech.ru}}
        }

\usepackage{graphicx}
\usepackage{subcaption}
\usepackage{hhline}
\usepackage{amsthm}
\usepackage{amsmath}
\usepackage{amssymb}
\usepackage{amsfonts}
\usepackage{verbatim}
\usepackage{multirow}
\usepackage{hhline}%
\usepackage{listings}
\usepackage{epstopdf}
\usepackage{multirow}
\usepackage{array}
\usepackage{subcaption}
\usepackage{blindtext}

\newcolumntype{C}[1]{>{\centering\let\newline\\\arraybackslash\hspace{0pt}}m{#1}}

\theoremstyle{plain}
\newtheorem{theorem}{Theorem}[section]

\newtheorem{lem}[theorem]{Lemma}

\theoremstyle{definition}

\theoremstyle{remark}

\newcommand{\diag }[0]{\operatorname{diag }}

\allowdisplaybreaks

\begin{document}
\maketitle
\maketitle
\begin{abstract}
Introduction of renewable generation leads to significant reduction of inertia in power system, which deteriorates the quality of  frequency control. This paper suggests a new control scheme utilizing controllable load to deal with low inertia systems. Optimization problem is formulated to minimize the system�s deviations from the last economically optimal operating point. The proposed scheme combines frequency control with congestion management and maintaining inter-area flows. The proposed distributed control scheme requires only local measurements and communication with neighbors or between buses participating in inter-area flows. Global asymptotic stability is proved for arbitrary network. Numerical simulations confirm that proposed algorithm can rebalance power and perform congestion management after disturbance with transient performance significantly improved in comparison with the traditional control scheme.
\end{abstract}

\section{INTRODUCTION}
The essence of power system control is to maintain system security at minimal cost. This paper deals with arguably two most important aspects of system security: keeping frequency within tight bounds around the nominal value while maintaining inter-area flows, and keeping power flows below line limits so that the lines are not overloaded (congestion management). Economic dispatch is enforced by running Security Constrained Optimal Power Flow (SCOPF). Typically (N-1) security standard is observed which means that no single contingency (i.e. a loss of a single element such as a line or a generator) should result in such redistribution of power flows that any line is overloaded.

Frequency control scheme consist of three components \cite{WW} - \cite{ADSJ}. Primary frequency control is aimed to counter initial frequency drop using only local frequency measurements. It is implemented using droop control of turbine governors, and operates at timescale of tens of seconds. The secondary frequency control called Automatic Generation Control (AGC) is centralized and uses integral controller in order to deliver frequency to its nominal value (50 or 60 HZ). Additionally it ensures that inter-area flows in the network remain unchanged.

This traditional control scheme was developed to operate on conventional generators, which have high inertia values due to the size of turbines.  As  a  result,  frequency  does  not change abruptly in case  of a large disturbance and frequency control has time to react. Introduction of renewable generation leads to significant reduction of inertia, thus reducing robustness of the existing control scheme \cite{LI}, \cite{ASD}. Because of this fact some countries, like Ireland , have established constraints on the instantaneous penetration of renewable generation.

One way to improve frequency response in low-inertia systems is by involving responsive loads \cite{CH} - \cite{SIF}. Including load side control has an advantage that, unlike generators, loads can react quickly  to control commands hence improving frequency control. However including a large number of controllable loads makes the centralized control more complicated. Consequently a number of distributed frequency control  schemes were recently developed \cite{ADSJ}, \cite{DSB} - \cite{BLD}. The control scheme proposed in this paper, first presented in \cite{Ol1}, is also distributed and can include load-side control but it works on estimation of the disturbance size rather than directly with frequency while minimizing the deviation from the last optimal operating point. This paper extends \cite{Ol1} by adding congestion management, maintaining inter-area flows and deals with infeasible cases, when the disturbance size or location does not allow to  do both frequency control and congestion management. In this case the control prioritises returning frequency to the nominal value, even if line or inter-area flows limits will stay violated. Only communication between neighbours and between buses, participating in inter-area flows is required. This reduces the amount of communication with System Operator and makes it possible to connect the loads on the plug-and-play basis, when new buses can be added to the system without adjusting control parameters of the other ones.

Introducing congestion management as a part of the frequency control makes it possible to switch economic dispatch from (N-1) security standard to (N-0) as a corrective control will be activated when a disturbance occurs to prevent line overloads. This will reduce the cost of economic dispatch.

We consider Power network model \cite{BH}, \cite{WW} - \cite{BV} that includes turbines and governors dynamics,  swing dynamics at generator buses and linearized power flows. Behaviour of turbine and governor are approximated by linear first order equations for control design. Simulations undertaken using 39-bus New England system demonstrated that the presented control performs frequency restoration at  a faster timescale than the AGC while doing congestion management at the same time.
The article is organized in the following way. In section I-C network model and control aims are described, idea of derivation is given in the section II and control itself  is given in III. Results of numerical experiments are given IV. Conclusion and final observations are given in V.

\section{The Power System Model}\label{psm}
\subsection{Notations}
Let $\mathbb{R}$ be the set of real numbers. For an arbitrary matrix (vector) $X$ its transpose is denoted by $X^T$. $\textbf{0}$ and $I$ are zero matrix and identity matrix of the corresponding size respectively, $0$ and $\rho$ are zero vector and vector of ones of the corresponding size respectively. $\diag(x_1,\dots,x_2)$ is diagonal matrix with elements $x_i$.
If $x$ and $y$ are vectors, then comparison operators ($<,>,\le,\ge$) are considered to be elementwise.
\subsection{Model description}\label{ps}
Classical generator model \cite{BH}, \cite{WWW} combined with generator speed governor \cite{KN} and turbine model is used.
The power transmission network \cite{BV} is described by a connected directed graph $(V,E)$,  where $V$ is the set of $n$ buses, $E$ is set of $m$ lines. The network consists of $g$ generators and $g$ load buses. We define as $G$ and $L$ sets of generator and load buses respectively. We fix bus voltage to be equal $1$ p.u., line flows are approximated with linear DC equations.
Dynamics of the power transmission network is defined by the system of linear differential algebraic equations \cite{MZL}, \cite{ZMLB}, \cite{ZL}, \cite{ZML}, which is given below.
\begin{subequations}\label{msys}
\begin{align}
M\dot\omega_G =& -D_G\omega_G-C_Gp+p^M + p^D_G\label{msys1}\\
0 =& -D_L\omega_L-C_Gp+p^C_L + p^D_L\label{msys2}\\
\dot p =& BC^T\omega\label{msys3}\\
T^M\dot p^M =& -p^M+\alpha,\label{msys4}\\
T^G\dot\alpha_i =& -\alpha_i+p^C_G,\label{msys5}
\end{align}
\end{subequations}
Additionally dependance on phase angles deviations will be used
\begin{equation}\label{msys_ph2}
p = BC^T\theta,
\end{equation}

Variables of the system have the following meanings:
\begin{itemize}
  \item $\omega=(\omega_1,\dots,\omega_n)^T$ is vector of deviations of bus frequencies from nominal value,
  \item $p=(p_1,\dots,p_m)^T$ is vector of deviations of line active power flows from their reference values,
  \item $p^M=(p^M_1,\dots,p^M_g)$ is vector of mechanic power injections at generators,
  \item $\alpha=(\alpha_1,\dots,\alpha_g)$ is vector of positions of valves,
  \item $p^C=(p^C_i,\dots,p^C_n)$ is vector of control values,
  \item $\theta = (\theta_1,\dots\theta_n)^T$ is vector of deviations of phase angles from reference value.
\end{itemize}

Here for any vector $n\in\mathbb{R}^n$ notations $x_G$ and $x_L$ define subvectors of components corresponding to load and generator buses respectively. 

Here variables stand for deviations from nominal values, for simplicity symbol $\triangle$ is omitted. Parameters of the system:
\begin{itemize}
    \item $M=\diag(M_1,\dots,M_g)$ are generators inertia constants,
    \item $D=\diag(d_i,\dots,d_n)$ are steam and mechanical damping of generators and frequency-dependent loads, $D_G$ is submatrix that consisting only $d_i,\;i\in G$, $C_L$ consists of $d_i,\;i\in L$. $M$ is diagonal matrix of $M_i,\;i\in G$, analogically we define $T^M$, $T^G$ and $W$.
    \item $C$ is an incidence matrix of the graph, $C_G$ is submatrix of $C$ constructed by rows $C_i,\;i\in G$, analogically we define $C_L$.
    \item $p^D=(p^D_1,\dots,p^D_n)$ is vector of unknown disturbances (assumed constant),
    \item $B=\diag(b_1,\dots,b_m)$ are line parameters that depend on line susceptances, voltage magnitudes and reference phase angles,
    \item $T^M=\diag(T^M_1,\dots,T^M_g)$ are time constants that characterize time delay in fluid dynamics in the turbine,
    \item $T^G=\diag(T^G_1,\dots,T^G_g)$ are time constants that characterize time delay in governor response.
\end{itemize}
Equations of the system:
\begin{itemize}
    \item Equations (\ref{msys1}) describe classical generator dynamics,
    \item Equations (\ref{msys2}) are power balance equations for load buses,
    \item Equations (\ref{msys3}) are equations of DC linearized power flows,
    \item Equations (\ref{msys4}) are turbine dynamics equations,
    \item Equations (\ref{msys5}) are governors dynamics equations.
\end{itemize}
Deviations of bus power injections from the nominal values are defined in the following way:
\begin{equation}\label{epe}
p^E = -Cp.
\end{equation}


Control bounds are given by
\begin{equation}\label{cl}
    \underline p^C\le p^C\le\overline p^C.
\end{equation}

Line bounds are given by
\begin{equation}\label{ll}
    \underline p\le p\le\overline p,
\end{equation}
\begin{equation}\label{ll2}
    \underline p\le p.
\end{equation}
to simplify control derivation and stability proof. They will be added to the final control description in the section \ref{c_desc}.
Finally inter-area flows constraints are given by equality constraints
\begin{equation}\label{al}
    \sum_{j=1}^ms^k_jp_j=\psi_k,\;k=\overline{1,\xi},
\end{equation}
where $\xi$ is number of inter-area flows of the system and variables $s^k_j\in\{-1,0,1\},\;j=\overline{1,m},\;k=\overline{1,\xi}$ are indicators whether line $j$ participates in the inter-area flow $k$, or not, $-1$ and $1$ are used to ensure correct direction of the power flow.
In matrix form this constraint has form
\begin{equation}\label{mal}
    Sp-\psi=0,
\end{equation}
where element $k,j$ of $S$ is equal to $s^k_j$.

The objective function, that we want to minimize within the presented control scheme is defined as the deviation from the last economic dispatch
\begin{equation}\label{op}
f(p^C)=\frac{1}{2}(p^C)^TWp^C,
\end{equation}
where $W\succ0$ is diagonal matrix of costs of flexible generators that participating in frequency control.

It is assumed, that information exchange can be done only between neighboring buses. As a result, control has to work in a distributed fashion.

\subsection{Control aims}\label{ps}

It is required to derive control vector function $p^C(t)$, so that it satisfies the following criteria:
\begin{enumerate}
    \item Control values must be always within control limits (\ref{cl}).
    \item Frequency deviations $\omega(t)$ must asymptotically converge to zero if that does not violate control limits.
    \item Control must perform congestion management and keep inter-area flows within limits in stationary point. In stationary point system must satisfy line bounds (\ref{ll}), (\ref{al}) if that does not violate previous items.
    \item Control must minimize $f(p^C)$, if all previous items are satisfied.
    \item Control on each bus can use only local information and information from neighbor buses.
\end{enumerate}

\section{Idea of control}\label{ssc}

If we omit infeasible case, then vectors $\omega$, $p$, $p^M$, $\alpha$, $p^C$ must minimize $f(p^C)$ under constraints (\ref{cl}), (\ref{ll}), (\ref{mal}) and also satisfy algebraic equations that we acquire from (\ref{msys}) by bringing all derivatives to zero. This system gives a set of algebraic equations than can be simplified to the following:
\begin{equation}\label{om_con}
  \omega=0,\;p^M=\alpha=p^C_G,
\end{equation}
\begin{equation}\label{be}
-Cp+p^C+p^D=0.
\end{equation}

From (\ref{om_con}) we can exclude $\omega$, $p^M$ and $\alpha$. Since we assume lossless lines, matrix of the system (\ref{msys}) is singular, and the system has infinite number of equilibriums with different $p$. Physically that means, that system can contain power flow cycles. To ensure, that we do not encounter this problem in the future, equations (\ref{msys_ph2}) are utilized and the following substitution is used
\begin{equation}\label{p_eta}
    BC^T\eta=p,
\end{equation}

Here vector of axillary variables $\eta$ is used instead of $\theta$ to avoid variables repetition in physical system and in the control scheme.
Then (\ref{be}) is equivalent to
\begin{equation}\label{ebe}
    -CBC^T\eta+p^C+p^D=0.
\end{equation}
The minimization of deviation from the last economic dispatch problem has form (\ref{op}), (\ref{cl}), (\ref{ll}), (\ref{mal}), (\ref{ebe}) and further will be referred as optimal control problem.

\subsection{Frequency control with no constraints}
Let us assume that inequality constraints (\ref{cl}), (\ref{ll}), (\ref{al}) are not present. Then Lagrange function \cite{B} has form
\begin{equation}\label{numfil4}
L_1(p,p^C,\lambda)=\frac{1}{2}(p^C)^TWp^C-\lambda^T(-CBC^T\eta+p^C+p^D)
\end{equation}
Its stationary point satisfies the following system:
\begin{subequations}\label{sys1s}
\begin{align}
&CBC^Tu=0,\\
&CBC^T\eta-W^{-1}u-p^D=0,\\
&Wu=p^C.
\end{align}
\end{subequations}
Consider the following differential system:
\begin{subequations}\label{dsys1}
\begin{align}
\dot\eta=&-CBC^Tu,\\
\dot u=&CBC^T\eta-W^{-1}u-p^D,\\
p^C=&Wu.
\end{align}
\end{subequations}
It is linear differential algebraic system which is stable and converges to the solution of (\ref{msys1}) thus giving the solution of optimal control problem.


\subsection{Frequency control with control limits}
Let us now assume, that only constraints (\ref{cl}) are present. Then slight change to system (\ref{dsys1}) can be applied:
\begin{subequations}\label{dsys2}
\begin{align}
\dot\eta&=-CBC^Tu,\\
p^C &= [W^{-1}u]^{\overline p^C}_{\underline p^C},\\
\dot u&=CBC^T\eta-p^C-p^D.
\end{align}
\end{subequations}
Here for vectors $x$, $\overline y$ and $\underline y$ of the same size, $[x]^{\overline y}_{\underline y}$ is vector function, such that $([x]^{\overline y}_{\underline y})_i=\min\{\overline y_i,\max\{\underline y_i,x_i\}\}$.

If $\frac{u_i}{w_i}$ is bigger than corresponding control limit $\overline p^C_i$, we can say, that this bus does not have ay control, however its disturbance is not $p^D_i$, but $p^D_i+\overline p^C_i$ and $u_i$ is just an axillary valuable used by other buses. As soon as $\frac{u_i}{w_i}$ becomes smaller, than control limit than we again have control of this bus, and disturbance on it is equal $p^D_i$. It is possible to introduce control limits as inequality constraints and add them to Lagrange function, but they can be violated in infeasible case, which is described in the next section. Introduction of variables $u_i$ guarantees that control limits will be satisfied, even if it is impossible to restor frequency or satisfy line limits.
\subsection{Infeasible case with control limits}\label{iccl}
When disturbance size is bigger, then control reserve
\begin{equation}\label{imb}
-\sum_{i=1}^np^D_i>\sum_{i=1}^n\overline p^C_i,
\end{equation}
then in the system (\ref{dsys3}) variables $u_i$ will grow infinitely. In order to counter this effect additional limits are added:
\begin{subequations}\label{dsys3}
\begin{align}
\dot\eta&=-CBC^Tu,\label{dsys31}\\
p^C &= [W^{-1}u]^{\overline p^C}_{\underline p^C},\\
\dot u&=\Gamma^u(CBC\eta-p^C-p^D).\\
\end{align}
\end{subequations}
where
\begin{equation}\label{kul}
\begin{gathered}
\Gamma^u=\diag(\gamma_1^u,\dots,\gamma_n^u),\\
\gamma_i^u=\left\{
             \begin{array}{cc}
               0, & \mbox{if }u_i=\overline K^u_i\mbox{ and }(CBC\eta)_i-p^C_i-p^D_i\ge0, \\
               0, & \mbox{if }u_i=\underline K^u_i\mbox{ and }(CBC\eta)_i-p^C_i-p^D_i\le0, \\
               1, & \mbox{otherwise}, \\
             \end{array}
           \right.\\
\overline K^u_i>\max\{\overline p^C_i,\;i=\overline{1,n}\},\;i\in\overline{1,n}.\\
\end{gathered}
\end{equation}
 The last inequality is important, since from equation (\ref{dsys31}) in equilibrium $u=\gamma\rho$, where $\gamma$ is some scalar, and in case when frequency cannot be restored, all control values must reach their upper bounds: $p^C=\overline p^C\le u$. This way frequency will be brought to the closest possible value to nominal.
\subsection{Frequency control with line limits}
Let us assume, that constraints (\ref{ll}) are present but (\ref{cl}) and (\ref{mal}) are not. Initially we will use (\ref{ll}) as exact line constraints, given by equalities: $p=\overline p$. Lagrange function for this problem has the following form:
\begin{equation}
\begin{gathered}
L_2(p,p^C,\eta,\lambda,\xi,\mu)=\frac{1}{2}(p^C)^TWp^C-\\-\lambda^T(-CBC\eta+p^C+p^D)-\overline\mu^T(BC^T\eta-\overline p).
\end{gathered}
\end{equation}
Let us present corresponding differential system as before:
\begin{subequations}\label{dsys5}
\begin{align}
\dot\eta &= -CBC^Tu-CB\overline\mu,\label{dsys51}\\
\dot u &= CBC^T\eta-W^{-1}u-p^D,\\
\dot{\overline\mu} &= \Gamma^{\overline\mu}(BC^T\eta - \overline p),\label{dsys53}\\
p^C=&Wu,
\end{align}
\end{subequations}
where
\begin{equation}
\begin{gathered}
\Gamma^{\overline \mu}=\diag(\gamma_1^{\overline \mu},\dots,\gamma_m^{\overline \mu}),\\
\gamma_i^{\overline \mu}=\left\{
             \begin{array}{cc}
               0, & \mbox{if }\overline\mu_j=0\mbox{ and }\hat BC^T\eta - \overline p\le0, \\
               0, & \mbox{if }\overline\mu_j=K^{\overline \mu}_j\mbox{ and }BC^T\eta - \overline p\ge0,\\
               1 & \mbox{otherwise}. \\
             \end{array}
           \right.
\end{gathered}
\end{equation}
Here equation  Vector $\overline\mu$ is non-negative at any point in time and $\overline\mu^T(p-\overline p)=0$, so equations (\ref{dsys53}) correspond to the complementary slackness conditions. This way the differential system will converge to solution of the optimal control problem with line limits given by inequalities (\ref{ll}). 
Similarly to section \ref{iccl}, if the feasible set is empty, variables $\mu_j$ from (\ref{dsys53}) will converge to $\infty$, therefore it is bounded by $K^\mu_j>0$. Benefit of such approach is in the fact, that if frequency restoration and congestion management cannot be performed at the same time. In this case than dual variables $\overline\mu_j$ will reach their limits $\overline K^\mu_j$. Then constraints (\ref{ll}) will no longer affect differential system (\ref{dsys5}) and it will be responsible for frequency restoration only.

\begin{subequations}
\begin{align}
\dot\eta &= -CBC^Tu-CB\overline\mu,\\
\dot u &= CBC^T\eta-W^{-1}u-p^D,\\
\dot{\underline\mu} &= \Gamma^{\underline\mu}(\underline p - BC^T\eta),\\
p^C&=Wu,
\end{align}
\end{subequations}
where
\begin{equation}
\begin{gathered}
\Gamma^{\underline \mu}=\diag(\gamma_1^{\underline \mu},\dots,\gamma_m^{\underline \mu}),\\
\gamma_i^{\underline \mu}=\left\{
             \begin{array}{cc}
               0, & \mbox{if }\underline\mu_j=0\mbox{ and }\underline p - BC^T\eta\le0, \\
               0, & \mbox{if }\underline\mu_j=K^{\overline \mu}_j\mbox{ and }\underline p - BC^T\eta\ge0,\\
               1 & \mbox{otherwise}. \\
             \end{array}
           \right.
\end{gathered}
\end{equation}

\subsection{Frequency control with inter-area flows limits}
Let us now assume, that only constraints (\ref{mal}) are present.
As in previous cases, 
corresponding differential system is given by
\begin{subequations}\label{dsys3}
\begin{align}
\dot\eta&=CBC^Tv-CBS^T\phi,\\
\dot u &= CBC\eta-W^{-1}u-p^D,\\
\dot\phi&= \Gamma^\phi(SBC^T\eta-\psi),\\
p^C=&Wu,
\end{align}
\end{subequations}
where
\begin{equation}
\begin{gathered}
\Gamma^\phi=\diag(\gamma_1^\phi,\dots,\gamma_\xi^\phi),\\
\gamma_i^\phi=\left\{
             \begin{array}{cc}
               0, & \mbox{if }\phi_k=\overline K^\phi_k\mbox{ and } (SBC^T\eta)_k-\psi_k\ge0, \\
               0, & \mbox{if }\phi_k=\underline K^\phi_k\mbox{ and }(SBC^T\eta)_k-\psi_k\le0, \\
               1, & \mbox{otherwise}. \\
             \end{array}\right.
\end{gathered}
\end{equation}
Here $\overline K^\phi>0>\underline K^\phi$. As in previous cases upper and lower limit on $\phi$ are present to ensure, that the system will remain stable even in the infeasible case.

\section{Control}\label{c_desc}
If we combine results of the previous section, final version of control will have the following form:

\begin{subequations}\label{cnt}
\begin{align}
\dot\eta &= CB\left(-C^Tu-\overline\mu+\underline\mu-S^T\phi\right),\label{cnt1}\\
\chi &= BC^T\eta,\label{cnt2}\\
p^C &= [W^{-1}u]^{\overline p^C}_{\underline p^C},\label{cnt3}\\
\dot u&=\Gamma^u(C\chi-p^C-p^D),\label{cnt4}\\
\dot{\overline\mu} &= \Gamma^{\overline\mu}(\chi - \overline p),\label{cnt5}\\
\dot{\underline\mu} &= \Gamma^{\underline\mu}(\underline p - \chi),\label{cnt6}\\
\dot\phi&= \Gamma^\phi(S\chi-\psi),\label{cnt7}
\end{align}
\end{subequations}
Here for stability purposes $K^u$ has to satisfy the following inequalities:
\begin{equation}\label{kk_l}
\max\{-\underline K^\phi,\overline K^\phi, K^{\underline\mu},K^{\overline\mu}\}\ll\min\{-\overline K^u,\underline K^u\},
\end{equation}
\begin{equation}\label{kc_l}
\underline K^u\ll\underline p^C,\;\overline p^C\ll\overline K^u.
\end{equation}

\subsection{Stability proof}
We define vectors $\eta^*,u^*,\overline\mu^*,\underline\mu^*,\phi^*$ as an equilibrium point of the system (\ref{cnt}). It satisfies the following system of algebraic equations:
\begin{subequations}\label{ep}
\begin{align}
0 &= CB\left(-C^Tu^*-\overline\mu^*+\underline\mu^*-S^T\phi^*\right),\label{ep1}\\
\chi^* &= BC^T\eta^*,\label{ep2}\\
(p^C)^* &= [W^{-1}u^*]^{\overline p^C}_{\underline p^C},\label{ep3}\\
0 &=\Gamma^{u^*}(C\chi^*-(p^C)^*-p^D),\label{ep4}\\
0 &= \Gamma^{\overline\mu^*}(\chi^* - \overline p),\label{ep5}\\
0 &= \Gamma^{\underline\mu^*}(\underline p - \chi^*),\label{ep6}\\
0 &= \Gamma^{\phi^*}(S\chi^*-\psi),\label{ep7}
\end{align}
\end{subequations}

\begin{lem}
Solution of (\ref{ep}) always exists.
\end{lem}

\begin{proof}
Firstly equation (\ref{ep1}) can be written as
\begin{equation}\label{u_slv}
CBC^Tu^* = CB\left(-\overline\mu^*+\underline\mu^*-S^T\phi^*\right).
\end{equation}
Here matrix $CBC^T$ is singular, however its only null space vector is vector of once.
Additionally sum of elements of right hand side vector is always $0$ due to multiplication by incidence matrix $C$ so it is never in the null space of $CBC^T$, therefore for any $\overline\mu^*$, $\underline\mu^*$ and $\phi^*$ there exists such $u^*$, that equation (\ref{ep1}) is satisfied.
Note that equation (\ref{u_slv}) has infinite amount of solutions in $u$ of the form $u^*+\rho\beta$, $\beta\in\mathbb{R}$ due to matrix $CBC^T$ being singular. Let us consider equation (\ref{ep2}).
For any $(p^C)^*$ such, that $\sum_{i=1}^n((p^C_i)^*+p^D_i)=0$, then exists such $\eta^*$, that (\ref{ep2}) is satisfied, since $C$ is an incidence matrix.
Otherwise there exists $u^*$ solution of (\ref{ep1}) such that at lease one of $u^*_i$ is equal to $\overline K^u_i$, therefore corresponding $\Gamma^{u^*}_i$ is equal $0$.
That excludes at least one equation from (\ref{ep2}) and one row from matrix $CB$, which means, that reduced matrix has full row rank and solution in $\eta^*$ always exists.
Let us now consider equations (\ref{ep4}).
Recall, that equations (\ref{ep1}) and (\ref{ep2}) can be satisfied for any $\overline\mu^*$, $\underline\mu^*$ and $\phi^*$.
If for some $j$ $\chi_j^*=\overline p_j$, than the equation is satisfied, otherwise $\overline\mu^*_j=K^{\overline\mu}_j$ and $\Gamma^{\overline\mu^*}_j=0$ and equation degenerates to $0=0$.
Similarly equations (\ref{ep5}) and (\ref{ep6}) are always satisfied.
\end{proof}

\begin{lem}
If all constraints (\ref{ll}), (\ref{cl}) and (\ref{mal}) are satisfied, stationary point $(\eta^*,u^*,\overline\mu^*,\underline\mu^*,\phi^*)$ deliver solution of the optimization problem (\ref{op}), (\ref{ebe}), (\ref{ll}), (\ref{cl}), (\ref{mal}).
\end{lem}

\begin{proof}
Lagrange function for this problem has form
\begin{equation}\label{l0}
\begin{gathered}
L_0(p^C,\eta,\lambda,\overline\sigma,\underline\sigma,\phi,\overline\mu,\underline\mu)=\frac{1}{2}(p^C)^TWp^C+\lambda^T(CBC^T\eta-p^C-p^D)+\\
+\overline\sigma^TW(p^C-\overline p^C)+\underline\sigma^TW(\underline p^C-p^C)+\phi^T(SBC^T\eta-\psi)+\\
+\overline\mu^T(BC^T\eta-\overline p)+\underline\mu^T(\underline p-BC^T\eta)\\
\end{gathered}
\end{equation}
KKT condition for this function has the following form:
\begin{subequations}\label{kkt}
\begin{align}
&Wp^C-\lambda+W\overline\sigma-W\underline\sigma=0,\label{kkt1}\\
&CB(C^T\lambda+\overline\mu-\underline\mu+S^T\phi)=0,\label{kkt2}\\
&CBC^T\eta-p^C-p^D=0,\label{kkt3}\\
&\overline\sigma_i(p^C_i - \overline p^C_i)=0,\;i=\overline{1,n},\label{kkt4}\\
&p^C - \overline p^C\le0,\label{kkt5}\\
&\underline\sigma_i(\underline p^C_i-p^C_i)=0,\;i=\overline{1,n},\label{kkt6}\\
&\underline p^C-p^C\le0,\label{kkt7}\\
&\overline\sigma\ge0,\;\underline\sigma\ge0,\label{kkt8}\\
&\overline\mu_j((BC^T\eta)_j-\overline p_j)=0,\;j=\overline{1,m},\label{kkt9}\\
&BC^T\eta-\overline p\le0,\label{kkt10}\\
&\underline\mu_j(\underline p_j-(BC^T\eta)_j)=0,\;j=\overline{1,m},\label{kkt11}\\
&\underline p-BC^T\eta\le0,\label{kkt12}\\
&\overline\mu\ge0,\;\underline\mu\ge0,\label{kkt13}\\
&SBC^T\eta-\psi=0.\label{kkt14}
\end{align}
\end{subequations}
If we set $u=p^C+\overline\sigma-\underline\sigma$, then equations (\ref{kkt1})-(\ref{kkt8}) are equivalent to (\ref{ep1})-(\ref{ep4}). Additionally equations (\ref{ep5}), (\ref{ep6}) represent complementary slackness (\ref{kkt9}), (\ref{kkt11}), positiveness of dual variables (\ref{kkt13}) and inequality constraints (\ref{kkt10}) and (\ref{kkt12}). Furthermore equations (\ref{ep7}) and (\ref{kkt14}) are equivalent.
\end{proof}

\begin{lem}
If power reserve is sufficient to compensate the disturbance ($\sum_{i=1}^n\underline p^C_i\le-\sum_{i=1}^n p^D_i\le\sum_{i=1}^n\overline p^C_i$), then in equilibrium
$$\sum_{i=1}^n (p^C_i)^*=-\sum_{i=1}^n p^D_i.$$
\end{lem}
\begin{proof}
If $\Gamma^u=I$, then sum of equations (\ref{ep4}) we get $\sum_{i=1}^n p^C_i=-\sum_{i=1}^n p^D_i$.
If there exists $i_0\in\overline{1,n}$ such that $\Gamma^(u^*)_{i_0} = 0$, then either $u^*_{i_0}=\overline K^u_{i_0}$ or $u^*_{i_0}=\underline K^u_{i_0}$. From (\ref{ep1}) $u=(CBC^T)^+CB(-\overline\mu^*+\underline\mu^*-S^T\phi^*)+\rho\beta$, $\beta\in\mathbb{R}$. Therefore
\begin{equation}\label{clmax}
|u_{i_0}^*-u^*_i|\le\|(CBC^T)^+CB(-\overline\mu^*+\underline\mu^*-S^T\phi^*)\|,\;i=\overline{1,n}.
\end{equation}
Here $0\le\overline\mu^*\le K^{\overline\mu}$, $0\le\underline\mu^*\le K^{\underline\mu}$, $\underline K^\phi\le\phi^*\le\overline K^\phi$, therefore form (\ref{kk_l}), (\ref{kc_l}) if $u_{i_0}^*=\underline K^u$, then $\max_i|u_{i_0}^*-u_i^*|<\underline p^C_i-\underline K^u_{i_0},\;i=\overline{1,n},$ and $u^*<\underline p^C.$
As a result, $\sum_{i=1}^n (p^C_i)^*<\sum_{i=1}^n p^D_i$ and $\sum_{i=1}^n (p^C)^*_i+\sum_{i=1}^n p^D_i<0$. Since $\sum_{i=1}^n(CBC\eta)_i-p^C_i-p^D_i=-\sum_{i=1}^n p^C_i-\sum_{i=1}^n p^D_i$, there exists $i_1$ such that $(CBC\eta)_{i_1}-p^C_{i_1}-p^D_{i_1}>0$, therefore $\Gamma^u_{i_1}=1$ and equation $i_1$ in (\ref{ep4}) is not equal $0$.
Similar contradiction can be shown for the case, when $u^*_{i_0}=\overline K^u_{i_0}$.
\end{proof}

\begin{lem}\label{stp}
System (\ref{cnt}) is stable.
\end{lem}
\begin{proof}
Let us introduce the following Lyapunov function \cite{D}
\begin{equation*}
\begin{gathered}
V^0(\eta,u,\overline\mu,\underline\mu,\phi)=\frac{1}{2}((\eta-\eta^*)^T(\eta-\eta^*)+\\
+(u-u^*)^T(u-u^*)+\\+(\overline\mu-\overline\mu^*)^T(\overline\mu-\overline\mu^*)+(\underline\mu-\underline\mu^*)^T(\underline\mu-\underline\mu^*)+\\
+(\phi-\phi^*)^T(\phi-\phi^*)).
\end{gathered}
\end{equation*}
Since $CB\left(-C^Tu^*-\overline\mu^*+\underline\mu^*+S^T\phi^*\right)=0$, $\dot V$ can be split into 4 sums:
\begin{equation*}
\begin{gathered}
\dot V^0(\eta,u,\overline\mu,\underline\mu,\phi)=\Upsilon_1(\eta,u,\overline\mu,\underline\mu,\phi)+\\
+\overline\Upsilon_2(\eta,u,\overline\mu,\underline\mu,\phi)+\underline\Upsilon_2(\eta,u,\overline\mu,\underline\mu,\phi)+\Upsilon_3(\eta,u,\overline\mu,\underline\mu,\phi),
\end{gathered}
\end{equation*}
where
\begin{equation*}
\begin{gathered}
\Upsilon_1(\eta,u,\overline\mu,\underline\mu,\phi)=(u^*-u)^T(C\chi-C\chi^*)-\\-(u^*-u)^T(C\chi-p^C-p^D),
\end{gathered}
\end{equation*}
\begin{equation*}
\begin{gathered}
\overline\Upsilon_2(\eta,u,\overline\mu,\underline\mu,\phi)=(\overline\mu^*-\overline\mu)^T(\chi-\chi^*)-\\-(\overline\mu^*-\overline\mu)^T\Gamma^{\overline\mu}(\chi-\overline p),
\end{gathered}
\end{equation*}
\begin{equation*}
\begin{gathered}
\underline\Upsilon_2(\eta,u,\overline\mu,\underline\mu,\phi)=(\underline\mu-\underline\mu^*)^T(\chi-\chi^*)-\\-(\underline\mu-\underline\mu^*)^T\Gamma^{\underline\mu}(\chi-\underline p),
\end{gathered}
\end{equation*}
\begin{equation*}
\begin{gathered}
\Upsilon_3(\eta,u,\overline\mu,\underline\mu,\phi)=(\phi^*-\phi)^T(S^T\chi-S^T\chi^*)-\\-(\phi^*-\phi)^T\Gamma^{\phi}(S^T\chi-\psi),
\end{gathered}
\end{equation*}

Let us show, that $\Upsilon_1(\eta,u,\overline\mu,\underline\mu,\phi)\le0$.
\begin{enumerate}
\item If power reserve is sufficient to compensate the disturbance ($\sum_{i=1}^n\underline p^C_i\le-\sum_{i=1}^n p^D_i\le\sum_{i=1}^n\overline p^C_i$), then in equilibrium $p^D=C\chi^*-(p^C)^*$. Therefore $\Upsilon_1(\eta,u,\overline\mu,\underline\mu,\phi)=-(u-u^*)^T\Gamma^u(p^C-(p^C)^*)+(u-u^*)^T(\Gamma^u-I)(C\chi-C\chi^*)$. Note that if $u_i\le u^*_i$ then $p^C_i\le (p^C_i)^*$ and vice versa ($i\in\overline{1,n}$), therefore the first quadratic form is non positive. Let us show  summand of the second quadratic form $(u_i-u^*_i)(\Gamma^u_i-1)((C\chi)_i-(C\chi^*)_i)$ is non positive for any $i=\overline{1,n}$.
    \begin{enumerate}
    \item If $u_i=\overline K^u_i$ and $(C\chi)_i-p^C_i-p^D_i\ge0$, then $u_i-u_i^*\ge0$, $\Gamma^u_i=0$, and $((C\chi)_i-(C\chi^*)_i)\ge p^C_i-(p^C_i)^*\ge0$, therefore the summand is non positive.
    \item If $u_i=\underline K^u_i$ and $(C\chi)_i-p^C_i-p^D_i\le0$, then $u_i-u_i^*\le0$, $\Gamma^u_i=0$, and $((C\chi)_i-(C\chi^*)_i)\le p_i^C-(p_i^C)^*\le0$, therefore the summand is non positive.
    \item Otherwise $\Gamma^u_i=1$ and the summand is equal $0$.
    \end{enumerate}
\item If $\sum_{i=1}^n p^D_i\ge\sum_{i=1}^n\overline p^C_i$, then $u^*=\overline K^u$, and $C\chi^*-(p^C)^*-p^D\ge0$. As before first quadratic form is non positive. Any summand $(u_i-u^*_i)(\Gamma^u_i-1)((C\chi)_i-(C\chi^*)_i)$ of the second quadratic form has the following properties:
    \begin{enumerate}
    \item If $u_i=\overline K^u_i=u^*_i$, then the summand is equal $0$.
    \item If $u_i=\underline K^u_i$ and $(C\chi)_i-p^C_i-p^D_i\le0$, then $u-u^*\le0$, $\Gamma^u_i=0$, and $((C\chi)_i-(C\chi^*)_i)\le p^C-(p^C)^*\le0$, therefore the summand is non positive.
    \item Otherwise $\Gamma^u_i=1$ and the summand is equal $0$.
    \end{enumerate}
\item If $\sum_{i=1}^n\underline p^C_i\ge\sum_{i=1}^n p^D_i$, then $u^*=\underline K^u$, and $C\chi^*-(p^C)^*-p^D\le0$. First quadratic form is non positive. Any summand $(u_i-u^*_i)(\Gamma^u_i-1)((C\chi)_i-(C\chi^*)_i)$ of the second quadratic form has the following properties:
    \begin{enumerate}
    \item If $u_i=\overline K^u_i$ and $(C\chi)_i-p^C_i-p^D_i\ge0$, then $u-u^*\ge0$, $\Gamma^u_i=0$, and$((C\chi)_i-(C\chi^*)_i)\ge p^C-(p^C)^*\ge0$, therefore the summand is equal $0$.
    \item If $u_i=\underline K^u_i=u^*_i$, then the summand is equal $0$.
    \item Otherwise $\Gamma^u_i=1$ and the summand is equal $0$.
    \end{enumerate}
\end{enumerate}

Let us show, that $\overline\Upsilon_2(\eta,u,\overline\mu,\underline\mu,\phi)\le0$. Here $\overline\Upsilon_2(\eta,u,\overline\mu,\underline\mu,\phi)=\sum_{j=1}^m(\overline \mu^*_j-\mu_j)(\chi_j-\chi^*_j)-(\overline\mu^*_j-\mu_j)\Gamma^{\overline\mu}_j(\chi_j-\overline p_j)=\sum_{j=1}^n\overline\upsilon_j$. We will show that each summand is non positive. For sum $j$ we have the following cases:
\begin{enumerate}
\item If $\chi_j<\overline p_j$ and $\overline\mu_j=0$, then $\Gamma^{\overline\mu}_j=0$ and $\overline\upsilon_j=\overline\mu^*_j(\chi-\chi^*)$.
    \begin{enumerate}
    \item If $\chi^*_j<\overline p_j$, then $\overline\mu^*_j=0$ and $\overline\upsilon_j=0$.
    \item If $\chi^*_j\ge\overline p_j>\chi_j$, then $\overline\upsilon_j\le0$, since by definition $\overline\mu^*\ge0$.
    \end{enumerate}
\item If $\chi_j>\overline p_j$ and $\overline\mu_j=K^{\overline\mu}_j$, then $\Gamma^{\overline\mu}_j=0$ and $\overline\upsilon_j=(\overline\mu^*_j-\overline\mu_j)(\chi-\chi^*)$.
    \begin{enumerate}
    \item If $\chi^*_j\le\overline p_j$, then $\overline\mu_j\ge\overline\mu^*_j$ and $\chi_j\ge\chi^*_j$, therefore $\overline\upsilon_j\le0$.
    \item If $\chi^*_j>\overline p_j$, then $\overline\mu_j=\overline\mu^*_j$, therefore $\overline\upsilon_j=0$.
    \end{enumerate}
\item In all other cases $\Gamma^{\overline\mu}_j=1$ and $\overline\upsilon_j=(\overline\mu^*_j-\overline\mu_j)(\overline p_j-\chi^*_j)$.
    \begin{enumerate}
    \item If $\chi_j^*<\overline p_j$, then $\overline\mu^*_j=0$ and $\overline\upsilon_j=-\overline\mu_j(\overline p_j-\chi^*_j)\le0$, since $\overline\mu_j\le0$ by definition.
    \item If $\chi_j^*=\overline p_j$, then $\overline\upsilon_j=0$.
    \item If $\chi_j^*>\overline p_j$, then $\overline\mu^*_j=K^{\overline\mu}_j\ge\overline\mu_j$, therefore $\overline\upsilon_j\le0$.
    \end{enumerate}
\end{enumerate}

Let us show, that $\underline\Upsilon_2(\eta,u,\overline\mu,\underline\mu,\phi)\le0$. Here $\underline\Upsilon_2(\eta,u,\overline\mu,\underline\mu,\phi)=\sum_{j=1}^m(\underline \mu^*_j-\mu_j)(\chi_j-\chi^*_j)-(\underline\mu^*_j-\mu_j)\Gamma^{\underline\mu}_j(\chi_j-\underline p_j)=\sum_{j=1}^n\underline\upsilon_j$. We will show that each summand is non positive. For sum $j$ we have the following cases:
\begin{enumerate}
\item If $\chi_j>\underline p_j$ and $\underline\mu_j=0$, then $\Gamma^{\underline\mu}_j=0$ and $\underline\upsilon_j=-\underline\mu^*_j(\chi-\chi^*)$.
    \begin{enumerate}
    \item If $\chi^*_j>\underline p_j$, then $\underline\mu^*_j=0$ and $\overline\upsilon_j=0$.
    \item If $\chi^*_j\le\underline p_j<\chi_j$, then $\underline\upsilon_j\le0$, since by definition $\underline\mu^*\ge0$.
    \end{enumerate}
\item If $\chi_j<\underline p_j$ and $\underline\mu_j=K^{\underline\mu}_j$, then $\Gamma^{\underline\mu}_j=0$ and $\underline\upsilon_j=(\underline\mu_j-\underline\mu^*_j)(\chi-\chi^*)$.
    \begin{enumerate}
    \item If $\chi^*_j\ge\underline p_j$, then $\underline\mu_j\ge\underline\mu^*_j$ and $\chi_j\le\chi^*_j$, therefore $\underline\upsilon_j\le0$.
    \item If $\chi^*_j<\underline p_j$, then $\underline\mu_j=\underline\mu^*_j$, therefore $\underline\upsilon_j=0$.
    \end{enumerate}
\item In all other cases $\Gamma^{\underline\mu}_j=1$ and $\underline\upsilon_j=(\underline\mu_j-\underline\mu^*_j)(\underline p_j-\chi^*_j)$.
    \begin{enumerate}
    \item If $\chi_j^*>\underline p_j$, then $\underline\mu^*_j=0$ and $\underline\upsilon_j=\overline\mu_j(\underline p_j-\chi^*_j)\le0$, since $\overline\mu_j\le0$ by definition.
    \item If $\chi_j^*=\underline p_j$, then $\underline\upsilon_j=0$.
    \item If $\chi_j^*<\underline p_j$, then $\underline\mu^*_j=K^{\underline\mu}_j\ge\underline\mu_j$, therefore $\underline\upsilon_j\le0$.
    \end{enumerate}
\end{enumerate}

Let us show, that $\Upsilon_3(\eta,u,\overline\mu,\underline\mu,\phi)\le0$. Here $\Upsilon_3(\eta,u,\overline\mu,\underline\mu,\phi)=\sum_{k=1}^\xi(\phi^*_k-\phi_k)((S^T\chi)_k-(S^T\chi^*)_k)-(\phi^*_k-\phi_k)\Gamma^\phi_k((S^T\chi)_k-\psi_k)=\sum_{k=1}^\xi\upsilon_k$. We will show that each summand is non positive. For sum $k$ we have the following cases:
\begin{enumerate}
\item If $(S^T\chi)_k\ge\psi_k$ and $\phi_k=\overline K^\phi_k$, then $\Gamma^\phi_k=0$ and $\upsilon_k=(\phi^*_k-\phi_k)((S^T\chi)_k-(S^T\chi^*)_k)$.
    \begin{enumerate}
    \item If $(S^T\chi^*)_k>\psi_k$, then $\phi^*_k=\overline K^\phi_k$ and $\upsilon_k=0$.
    \item If $(S^T\chi^*)_k<\psi_k$, then $\psi^*_k=\underline K^\phi_k$ and $\upsilon_k\le0$, since $\phi^*_k<\phi_k$.
    \item If $(S^T\chi^*)_k=\psi_k\le(S^T\chi)_k$, then $\upsilon_k\le0$, since $\phi^*_k<\phi_k$.
    \end{enumerate}
\item If $(S^T\chi)_k\le\psi_k$ and $\phi_k=\underline K^\phi_k$, then $\Gamma^\phi_k=0$ and $\upsilon_k=(\phi^*_k-\phi_k)((S^T\chi)_k-(S^T\chi^*)_k)$.
    \begin{enumerate}
    \item If $(S^T\chi^*)_k>\psi_k$, then $\phi^*_k=\overline K^\phi_k$ and $\upsilon_k\le0$, since $\phi^*_k>\phi_k$.
    \item If $(S^T\chi^*)_k<\psi_k$, then $\phi^*_k=\underline K^\phi_k$ and $\upsilon_k=0$.
    \item If $(S^T\chi^*)_k=\psi_k\ge(S^T\chi)_k$, then $\upsilon_k\le0$, since $\phi^*_k>\phi_k$.
    \end{enumerate}
\item In all other cases $\Gamma^\phi_k=1$ and $\upsilon_k=(\phi^*_k-\phi_k)(\psi_k-(S^T\chi^*)_k)$.
    \begin{enumerate}
    \item If $(S^T\chi^*)_k>\psi_k$, then $\phi^*_k=\overline K^\phi_k$ and $\upsilon_k\le0$, since $\phi^*_k\ge\phi_k$.
    \item If $(S^T\chi^*)_k<\psi_k$, then $\psi^*_k=\underline K^\phi_k$ and $\upsilon_k\le0$, since $\phi^*_k\le\phi_k$.
    \item If $(S^T\chi^*)_k=\psi_k\ge(S^T\chi)_k$, then $\upsilon_k\le0$.
    \end{enumerate}
\end{enumerate}
\end{proof}
\begin{lem}\label{asp}
In system (\ref{cnt}) control variables $p^C$ converge to their equilibrium values:
$$\lim_{t\rightarrow\infty}p^C_i(t)=(p^C_i)^*,\;i=\overline{1,n}.$$
\end{lem}
\begin{proof}
We have $\dot V^0(\eta,u,\overline\mu,\underline\mu,\phi)=0$ only if $\Upsilon^1(\eta,u,\overline\mu,\underline\mu,\phi)=0$ but $\Upsilon^1(\eta,u,\overline\mu,\underline\mu,\phi)\le-(u-u^*)^T(p^C-(p^C)^*)$. Therefore it can be possible only when $u_i=u^*_i,\;i=\overline{1,n}$, which gives the stationary point, or when both $u_i$ and $u^*_i$ either bigger than $\overline p^C_i$ or smaller then $\underline p^C_i$, but then $p^C_i=(p^C_i)^*=\overline p^C_i$ or $p^C_i=(p^C_i)^*=\underline p^C_i$ respectively.
\end{proof}



\subsection{Disturbance estimation}
Differential system (\ref{cnt}) requires usage of disturbance vector, which is unknown. From (\ref{msys1}), (\ref{msys2}), (\ref{epe}) vector of disturbances can be estimated as is shown below:
\begin{equation*}
\begin{gathered}
p^D_G=M\dot\omega_G+D_G\omega_G+p^E_G-p^M,\\
p^D_L=D_L\omega_L+p^E_L-p^C_L.
\end{gathered}
\end{equation*}
Here deviations of frequencies $\omega$ and deviations of electrical power $p^E$ can be measured directly. Mechanical power injections $p^M$ are approximated by $\tilde p^M$ given by the following equation:
\begin{subequations}\label{lag}
\begin{align}
  \tilde T^M\dot{\tilde p}^M&=-\tilde p^M+\tilde\alpha,\\
  \tilde T^G\dot{\tilde\alpha}&=-\tilde\alpha+p^C.
\end{align}
\end{subequations}
Therefore equation (\ref{cnt4}) is replaced by
\begin{equation}\label{cnt_new}
  \dot u=\Gamma^u(C\chi-p^C-q),
\end{equation}
where $q_G=M\dot\omega_G+D_G\omega_G+p^E_G-\tilde p^M$, $q_L=D_L\omega_L+p^E_L-p^C_L$. The control scheme (\ref{cnt1})-(\ref{cnt3}), (\ref{cnt_new}), (\ref{cnt5})-(\ref{cnt7}) is shown in Figure \ref{CS_BD}. Here block B1 is describes mechanical power approximation, block B2 describes disturbance estimation $q$, block B3 describes control signal derivation, block B4 describes congestion management, block B5 describes control over inter-area flows.

\begin{figure}[!t]
      \centering
      \includegraphics[clip,scale=0.55,trim=170 0 170 0]{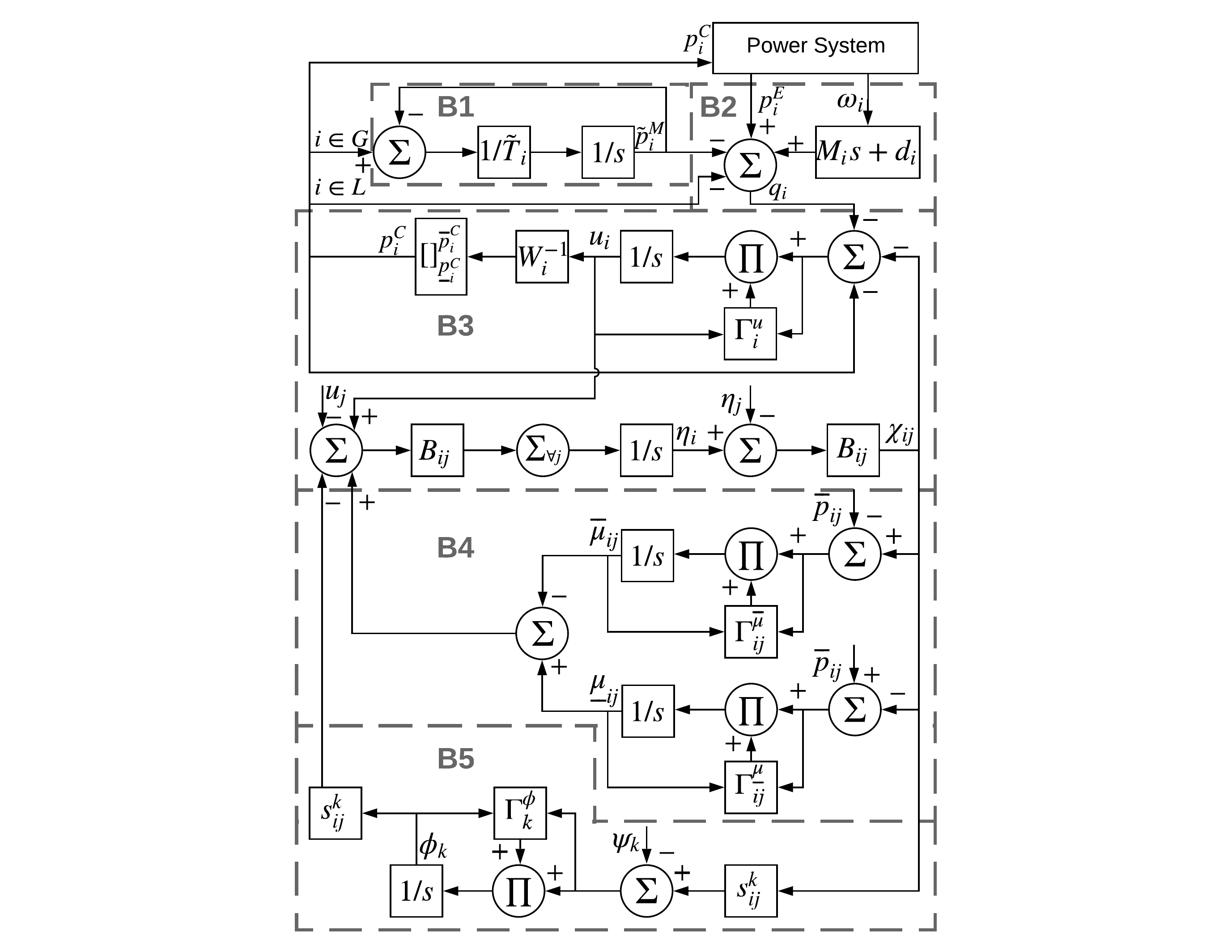}
      \caption{Control scheme block diagram}\label{CS_BD}
\end{figure}

With this change system (\ref{msys}), (\ref{cnt1})-(\ref{cnt3}), (\ref{cnt_new}), (\ref{cnt5})-(\ref{cnt7}) can be decoupled. Stability of (\ref{msys4}), (\ref{msys5}), (\ref{cnt1})-(\ref{cnt3}), (\ref{cnt_new}), (\ref{cnt5})-(\ref{cnt7}) leads to the stability of the entire system.

\begin{lem}
System (\ref{msys4}), (\ref{msys5}), (\ref{cnt1})-(\ref{cnt3}), (\ref{cnt_new}), (\ref{cnt5})-(\ref{cnt7}) is stable.
\end{lem}
\begin{proof}
The stability proof is similar to proof of lemma \ref{stp}. Let us introduce lagrange function
\begin{equation*}
\begin{gathered}
V^1(\eta,u,\overline\mu,\underline\mu,\phi,p^M,v,\tilde p^M)=\\=\frac{1}{2}((\eta-\eta^*)^T(\eta-\eta^*)+\\
+(u-u^*)^T(u-u^*)+\\+(\overline\mu-\overline\mu^*)^T(\overline\mu-\overline\mu^*)+(\underline\mu-\underline\mu^*)^T(\underline\mu-\underline\mu^*)+\\
+(\phi-\phi^*)^T(\phi-\phi^*)+(p^M-(p^M)^*)^TT^M(p^M-(p^M)^*)+\\+(v-v^*)^TT^G(v-v^*)+(\tilde p^M-(\tilde p^M)^*)^T\tilde T(\tilde p^M-(\tilde p^M)^*)).
\end{gathered}
\end{equation*}
As before we separate its derivative into 4 components
\begin{equation*}
\begin{gathered}
\dot V^1(\eta,u,\overline\mu,\underline\mu,\phi,p^M,v,\tilde p^M)=\Upsilon_1(\eta,u,\overline\mu,\underline\mu,\phi)+\\
+\overline\Upsilon_2(\eta,u,\overline\mu,\underline\mu,\phi)+\underline\Upsilon_2(\eta,u,\overline\mu,\underline\mu,\phi)+\Upsilon_3(\eta,u,\overline\mu,\underline\mu,\phi)+\\
+\Upsilon_4(\eta,u,\overline\mu,\underline\mu,\phi,p^M,v,\tilde p^M).
\end{gathered}
\end{equation*}
As is shown in Lemma \ref{stp}, $\underline\Upsilon_2(\eta,u,\overline\mu,\underline\mu,\phi)\le0$, $\overline\Upsilon_2(\eta,u,\overline\mu,\underline\mu,\phi)\le0$, $\Upsilon_3(\eta,u,\overline\mu,\underline\mu,\phi)\le0$ and $\Upsilon_1(\eta,u,\overline\mu,\underline\mu,\phi)\le(u-u^*)^T((p^C)^*-p^C)\le(p^C-(p^C)^*)^T((p^C)^*-p^C)$. Therefore
$$\Upsilon_1(\eta,u,\overline\mu,\underline\mu,\phi)+\Upsilon_4(\eta,u,\overline\mu,\underline\mu,\phi,p^M,v,\tilde p^M)\le y^TQy,$$
where $y=(p^C-(p^C)^*,p^M-(p^M)^*,v-v^*,\tilde p^M-(\tilde p^M)^*)$ and
$$Q=\left(
      \begin{array}{ccccc}
        -I & \Gamma^u & \textbf{0} & -\Gamma^u \\
        \textbf{0} & -I & I & \textbf{0} \\
        I & \textbf{0} & -I & \textbf{0}\\
        I & \textbf{0} & \textbf{0} & -I  \\
      \end{array}
    \right).
$$
This matrix is equivalent to block diagonal matrix with $|L|$ 5 by 5 blocks, of the following form
$$\frac{1}{2}\left(
      \begin{array}{rrrrr}
        -2 &  1 &  1 &  0 \\
         1 & -2 &  1 &  0 \\
         1 &  1 & -2 &  0 \\
         0 &  0 &  0 & -2 \\
      \end{array}
    \right) \mbox{ or }
    \frac{1}{2}\left(
      \begin{array}{rrrrr}
        -2 &  0 &  1 &  1 \\
         0 & -2 &  1 &  0 \\
         1 &  1 & -2 &  0 \\
         1 &  0 &  0 & -2 \\
      \end{array}
    \right)
$$
depending on $\Gamma^u_i$ being $1$ or $0$ respectively, $i\in L$. This blocks are symmetric and diagonally dominant thus negative semi-definite and as a result
$$\Upsilon_1(\eta,u,\overline\mu,\underline\mu,\phi)+\Upsilon_4(\eta,u,\overline\mu,\underline\mu,\phi,p^M,v,\tilde p^M)\le y^TQy\le0$$
and
$$\dot V^1(\eta,u,\overline\mu,\underline\mu,\phi,p^M,v,\tilde p^M)\le0.$$
\end{proof}
\begin{lem}
In system  (\ref{msys4}), (\ref{msys5}), (\ref{cnt1})-(\ref{cnt3}), (\ref{cnt_new}), (\ref{cnt5})-(\ref{cnt7}) control variables $p^C$ converge to their equilibrium values:
$$\lim_{t\rightarrow\infty}p^C_i(t)=(p^C_i)^*,\;i=\overline{1,n}.$$
\end{lem}
\begin{proof}
Proof is similar to the proof of the Lemma \ref{asp}
\end{proof}

\subsection{Communication constraints}
Note that differential system requires only multiplications by diagonal matrices, which can be done locally, by matrices $C$ or $C^T$ which by construction require only knowledge from neighbors and by matrix $S$ which requires only information sharing between buses participating in inter-area flows. As a result this scheme can be used in decentralized way.
\section{Numerical results}\label{nm}
Numerical experiments were undertaken using the New England System, Figure \ref{NE_sys}. Data from Power System Toolbox \cite{CR} was used. In order to ensure a fair comparison with AGC, only generator control was considered. Including load-side control would improve the transient performance of the controller. Additionally non-linear power flows are used \cite{MBB}, as well as single reheat tandem-compound steam turbine model \cite{KN} instead of first order equation (\ref{msys4}) for turbine dynamics.
\begin{figure}
      \centering
      \includegraphics[width=4in, trim=0 1 0 150]{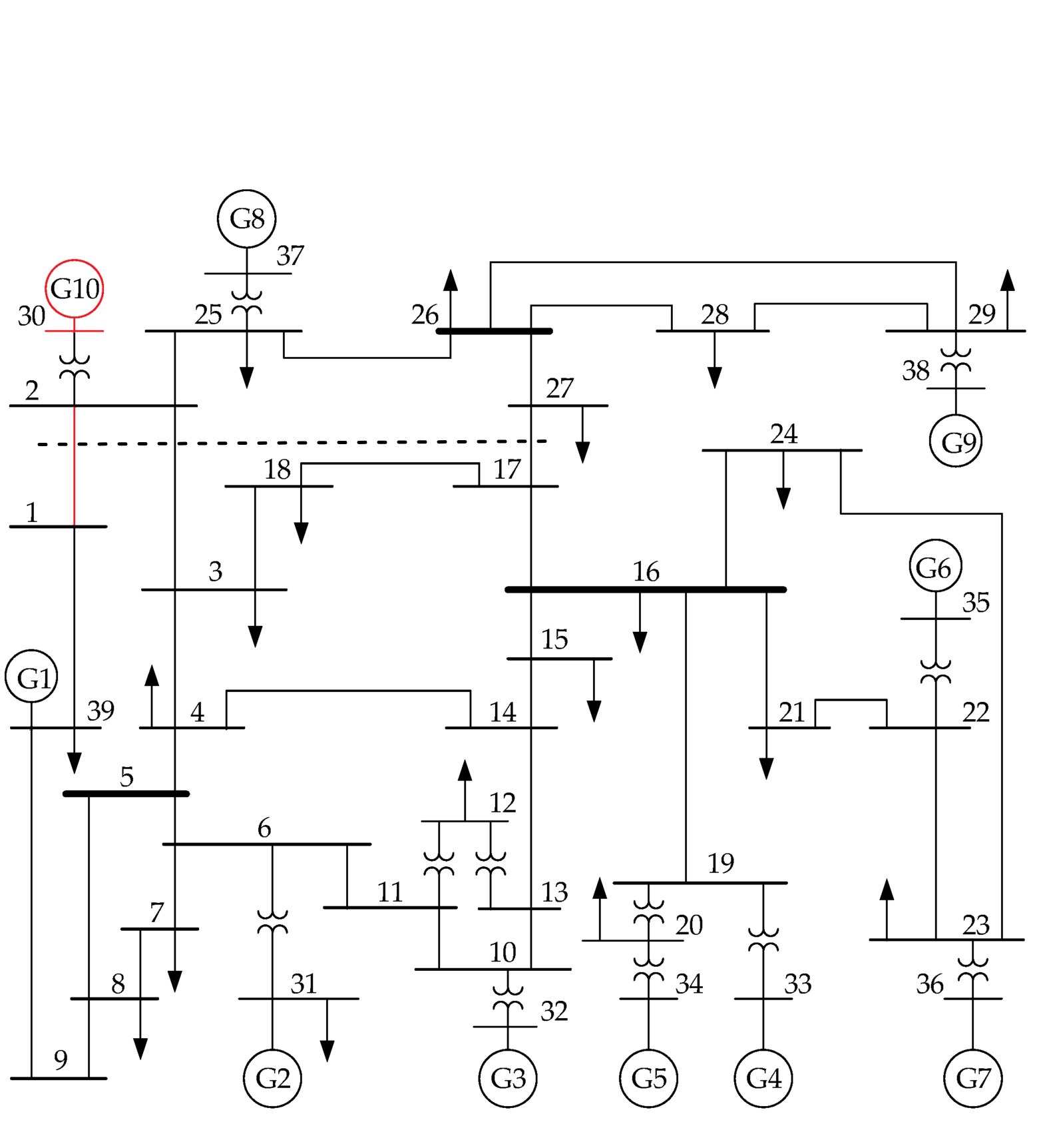}
      \caption{New England System.}\label{NE_sys}
      \label{figurelabel}
\end{figure}

Disturbance of 1 p.u. appears on the 10th generator (bus 30). System divided into two areas by lines (1,2), (2,3) and (17,27) sum of power flows deviations on this lines must be zero in equilibrium. In addition power flow on the line (1,2) must not exceed 0.1 p.u.

Figure \ref{freq} represents the system�s frequencies behavior under droop control with AGC (Figure \ref{freq_agc}) and under the proposed control (Figure \ref{freq_c}). It can be seen that control (\ref{cnt}) improves both the primary frequency control, as the frequency nadir (maximum frequency drop) is smaller, and the secondary frequency control as the settling time is smaller.

Power flows on line (1,2) and inter-area flows are given in Figures \ref{pf} and \ref{iaf} respectively. It can be seen, that unlike traditional control, when secure OPF is solved every 5 minutes or 1 hour, here congestion management is performed within the timescale of Primary Frequency control.

\begin{figure}
\centering
\begin{subfigure}[b]{0.4\textwidth}
\includegraphics[clip,width=3in,trim=90 250 100 250]{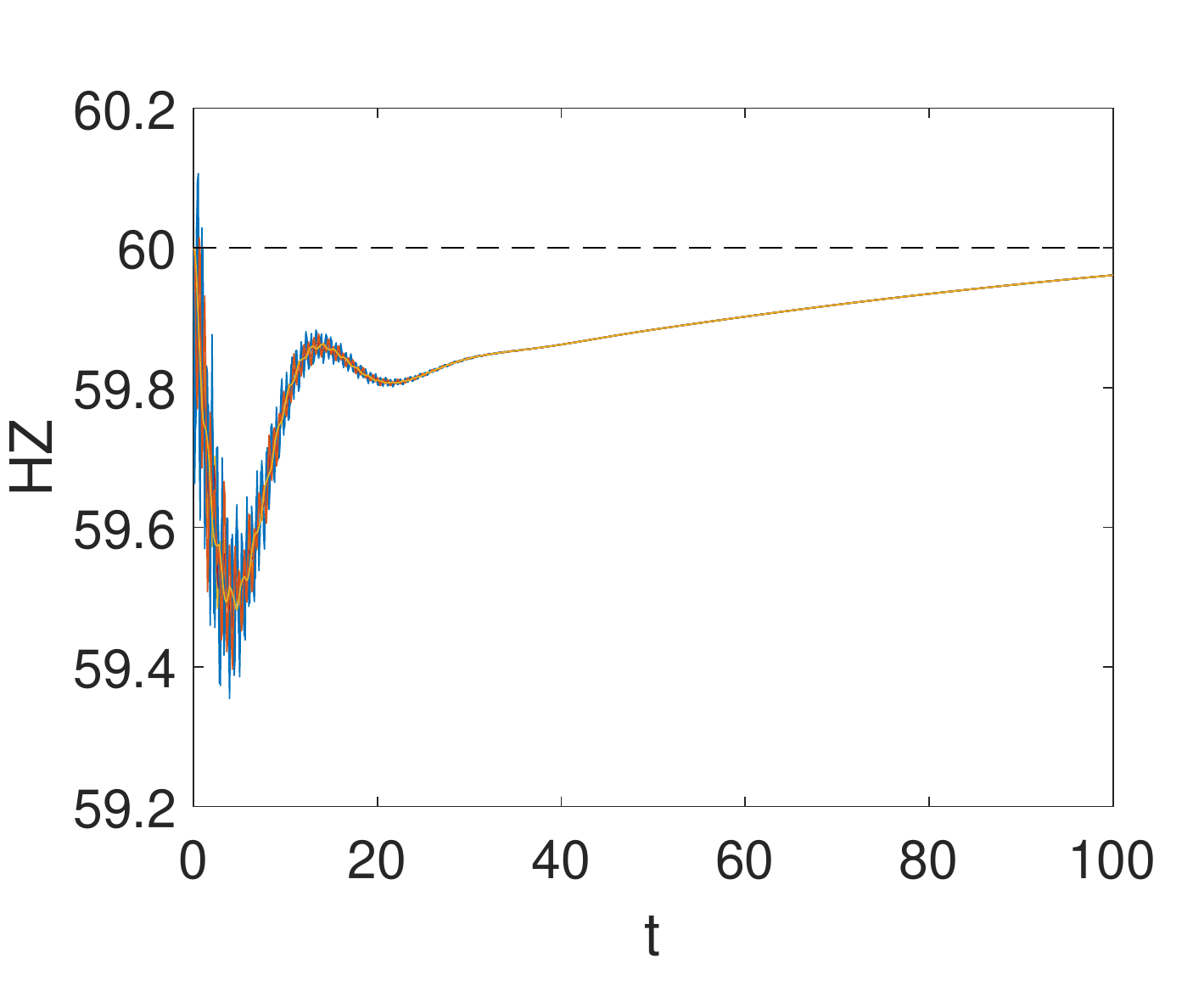}
\caption{Droop control and AGC.}\label{freq_agc}
\end{subfigure}
\begin{subfigure}[b]{0.4\textwidth}
\includegraphics[clip,width=3in,trim=90 250 100 250]{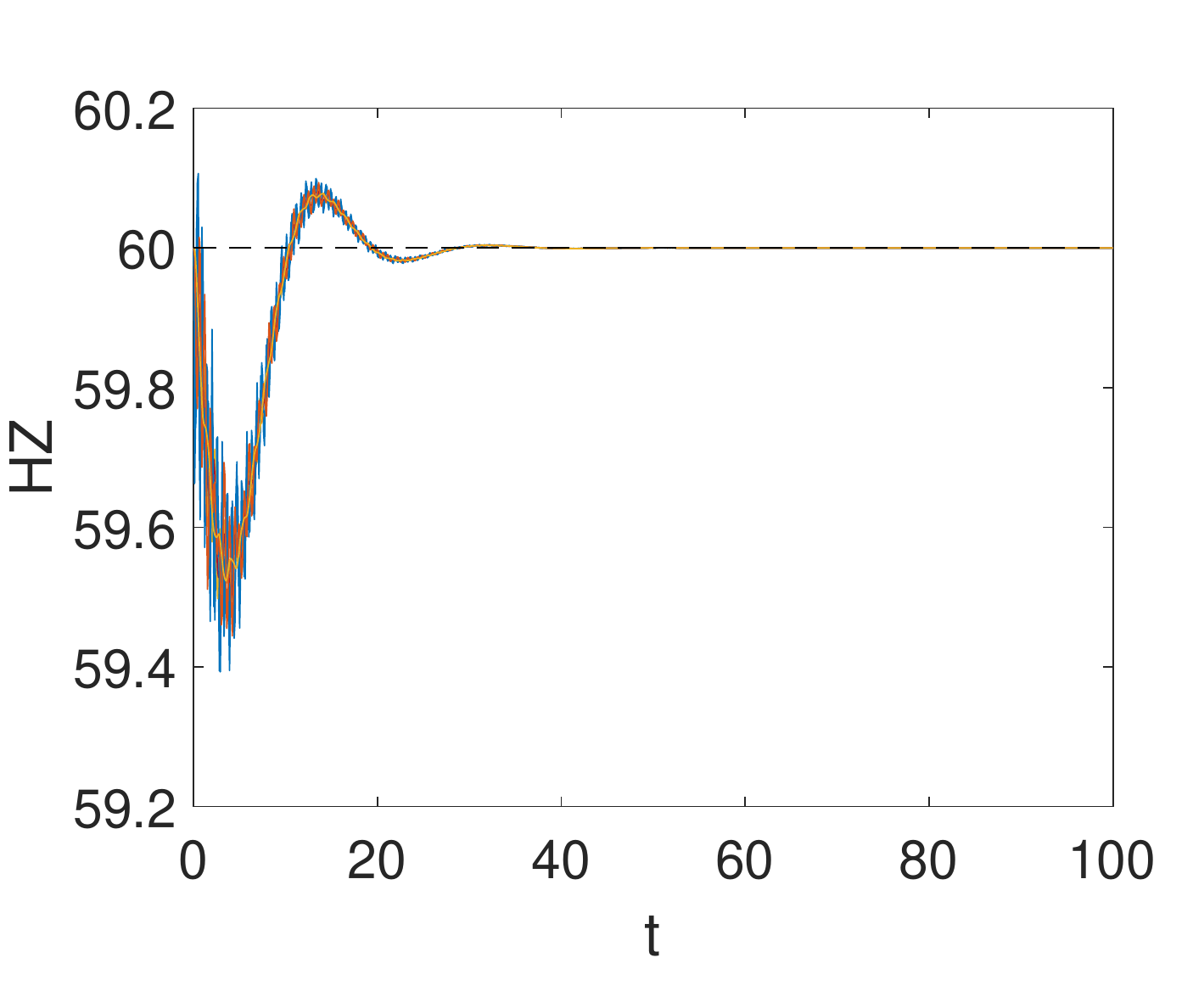}
\caption{Control (\ref{cnt}).}\label{freq_c}
\end{subfigure}
\caption{Frequencies of generator buses.}
\label{freq}
\end{figure}

\begin{figure}
\centering
\begin{subfigure}[b]{0.4\textwidth}
\includegraphics[clip,width=3in,trim=90 250 100 250]{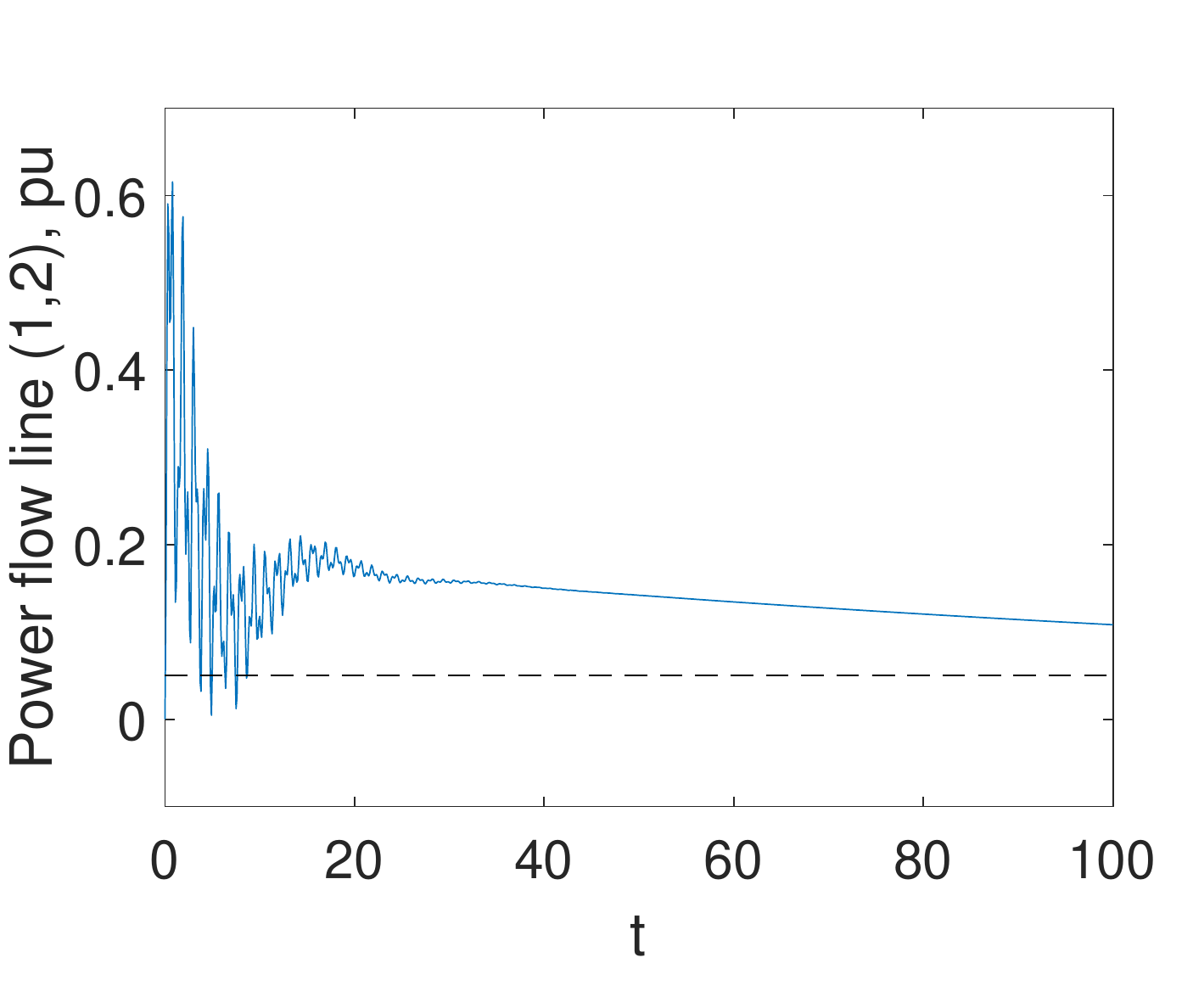}
\caption{Droop control and AGC.}\label{pf_agc}
\end{subfigure}
\begin{subfigure}[b]{0.4\textwidth}
\includegraphics[clip,width=3in,trim=90 250 100 250]{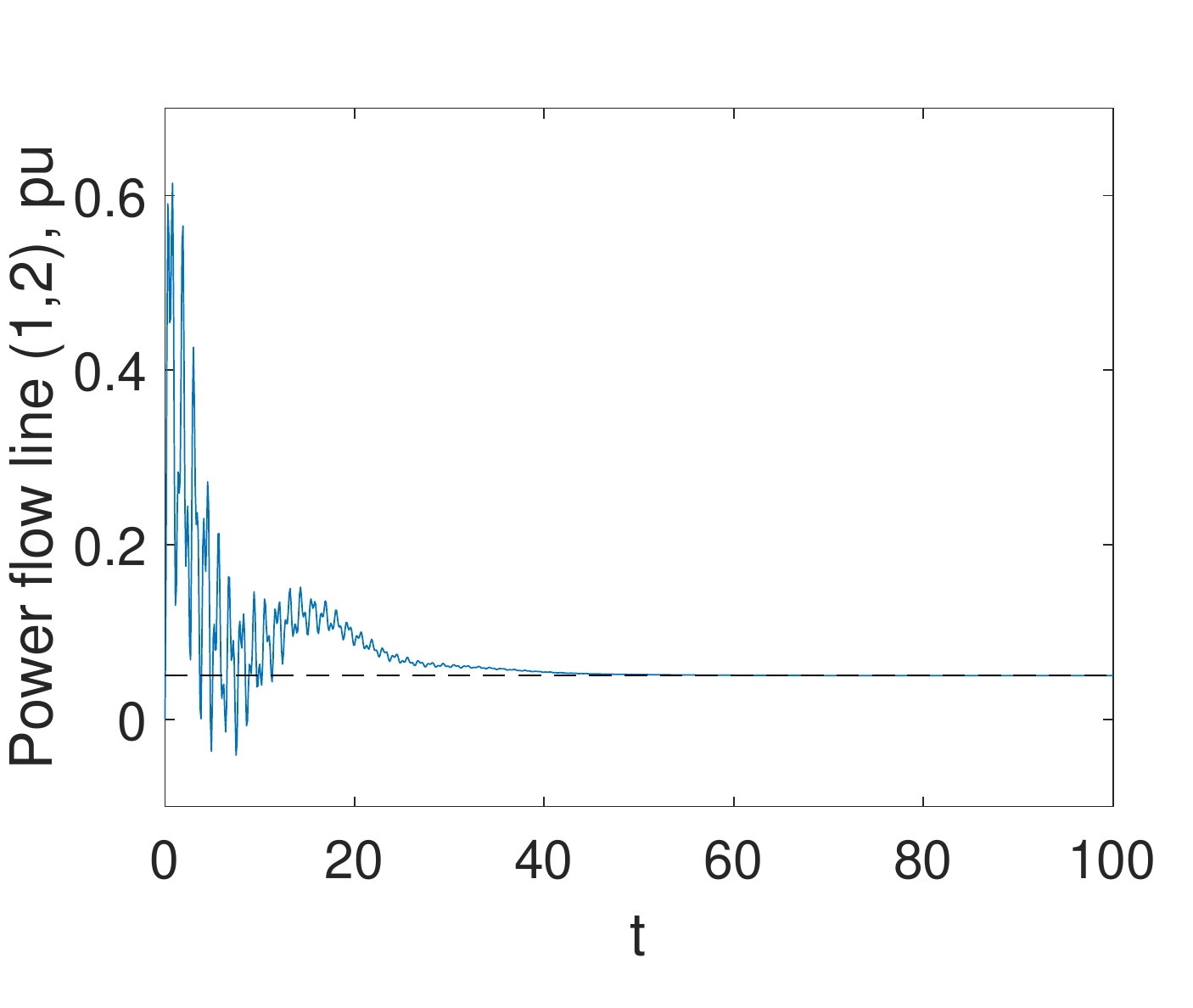}
\caption{Control (\ref{cnt}).}\label{pf_c}
\end{subfigure}
\caption{Deviations of power flow on the line (1,2).}
\label{pf}
\end{figure}

\begin{figure}
\centering
\begin{subfigure}[b]{0.4\textwidth}
\includegraphics[clip,width=3in,trim=90 250 100 250]{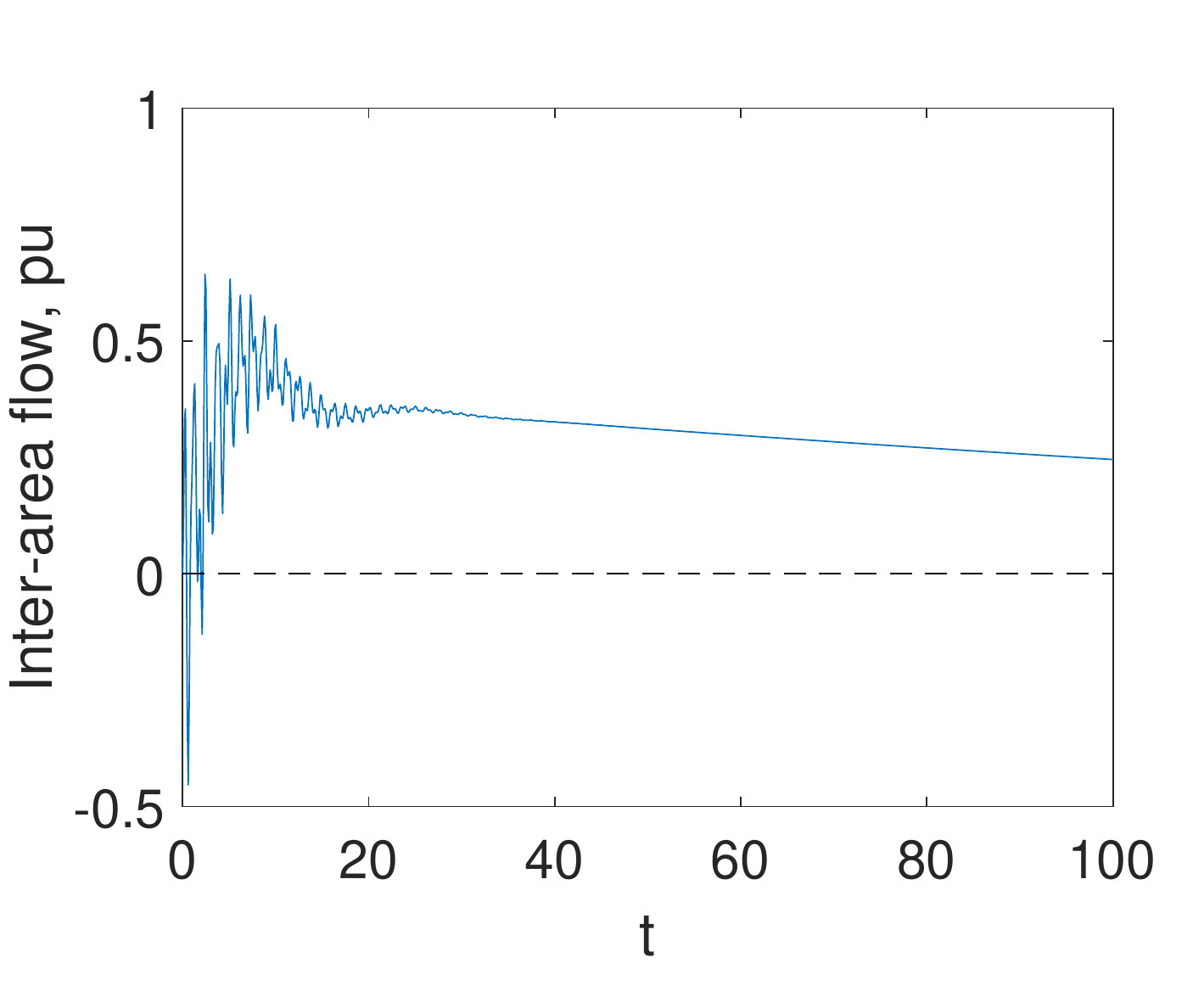}
\caption{Droop control and AGC.}\label{iaf_agc}
\end{subfigure}
\begin{subfigure}[b]{0.4\textwidth}
\includegraphics[clip,width=3in,trim=90 250 100 250]{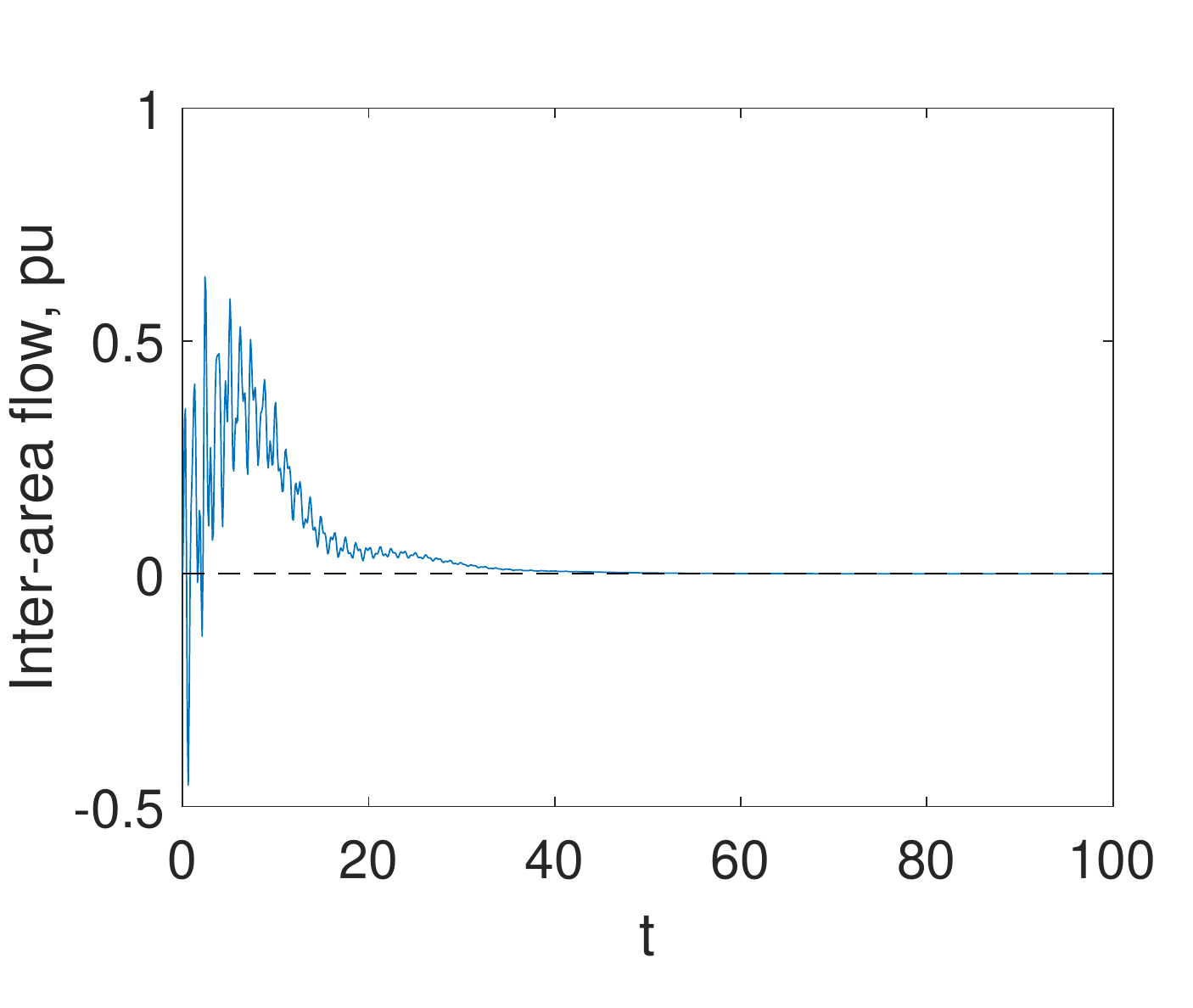}
\caption{Control (\ref{cnt}).}\label{iaf_c}
\end{subfigure}
\caption{Deviations of inter-area flow.}
\label{iaf}
\end{figure}

\section{Conclusion}\label{conc}
The control scheme, presented in this paper is designed as a replacement for AGC and  is aimed to provide faster frequency restoration and reduce frequency oscillations. It also performs congestion management in the timescale of tens of seconds hence making it possible to move from preventive (N-1) dispatch to (N-0) dispatch with corrective actions taken when a disturbance appears.
This will reduce the dispatch cost.
Power limits and inter-area flows constraints will be satisfied only if that des not affect frequency restoration, since their violation is not as dangerous as frequency oscillations. If all constraints are satisfied, then control delivers system to the closest to the last economic dispatch state in order to reduce generation cost.

Controllable loads are utilized as their responses to control signal are compared much quicker than those due compared to conventional generation and may therefore be used in low-inertia systems. Load side control greatly increases number of controllable busses. In order to ensure its control reliability in such circumstances, communication is limited to only communication between neighbours and between buses, which participate in inter-area flows. This also allows addition of new buses to the system without any communication with System Operator.

In the future we hope to ensure reliability of such control scheme in case of more complicated model, that includes nonlinear line flows and higher order generator equations.


\begin{thebibliography}{99}
\bibitem{WW}
A. J. Wood and B. F. Wollenberg,
{\it Power Generation, Operation, and Control}, 2nd ed.
NJ, US: John Wiley \& Sons, Inc.,
1996.
\bibitem{BH}
A. R. Bergen and D. J. Hill,
A Structure Preserving Model for Power System Stability Analysis,
{\it IEEE Transactions on Power Apparatus and Systems},
vol. PAS-100, no. 1,
pp 25-35,
1981.
\bibitem{ADSJ}
M. Andreasson, D. V. Dimarogonas, H. Sandberg, K. H. Johansson,
Distributed PI-Control with Applications to Power Systems Frequency Control,
{\it arXiv preprint arXiv:1311.0116v2},
2014.
\bibitem{LI}
Q. Liu, M. D. Ili\'c,
Enhanced automatic generation control (E-AGC) for
future electric energy systems,
{\it Proceedings of IEEE Power and Energy Society General Meeting},
San Diego, CA, USA, 2012.
\bibitem{ASD}
D. Apostolopoulou, P. W. Sauer, A. D. Dom\'inguez-Garc\'ia,
Balancing authority area model and its application to the design of adaptive AGC systems, {\it IEEE Transactions on Power Systems}, vol. 31, issue 5, 2016.
\bibitem{DSB}
F. D\"orfler, J. Simpson-Porco, F. Bullo,
Breaking the hierarchy: Distributed control and economic optimality in microgrids,
{\it IEEE Transactions on Control of Network Systems}, vol. 3, issue 3, 2016.
\bibitem{JLB}
A. Joki\'c, M. Lazar, P. P. van den Bosch,
Real-time control of power systems using nodal prices,
{\it International Journal of Electrical Power and Energy Systems},
vol. 31, issue 9, 2009.
\bibitem{ZTLL}
C. Zhao, U. Topcu, N. Li, S. H. Low,
Design and stability of load-side pri-
mary frequency control in power systems,
{\it IEEE Transactions on Automatic Control}, vol. 59, issue 5, 2014.
\bibitem{MZL}
E. Mallada, C. Zhao, S. Low,
Optimal load-side control for frequency regulation in smart grids,
{\it arXiv preprint arXiv:1410.2931v3},
2015.
\bibitem{ZMLB}
C. Zhao, E. Mallada, S. Low; J. Bialek,
A unified framework for frequency control and congestion management,
{\it 2016 Power Systems Computation Conference (PSCC)},
pp. 1-7,
2016.
\bibitem{LZC}
N. Li, C. Zhao, L. Chen,
Connecting automatic generation control and economic dispatch from an optimization view,
{\it IEEE Transactions on Control of Network Systems}, vol. 3, issue 3, 2016.
\bibitem{TBP}
S. Trip, M. B\"urger, C. De Persis,
An internal model approach to (optimal) frequency regulation in power grids with time-varying voltages,
{\it Automatica}, vol. 620, issue 64, 2016.
\bibitem{SPS}
T. Stegink, C. De Persis, A. van der Schaft,
A unifying energy-based approach to stability of power grids with market dynamics,
{\it IEEE Transactions on Automatic Control}, vol. 62, issue 6, 2017.
\bibitem{ZP}
X. Zhang, A. Papachristodoulou, A real-time control framework for smart power networks: Design methodology and stability,
{\it Automatica}, vol. 58, 2015.
\bibitem{KMDL}
A. Kasis, N. Monshizadeh, E. Devane, I. Lestas,
Stability and optimality of distributed secondary frequency control schemes in power networks,
arXiv:1703.00532.
\bibitem{I}
M. D. Ili\'c,
{\it From hierarchical to open access electric power systems},
Proceedings of the IEEE, vol. 95, issue 5, 2007.
\bibitem{BLD}
A. Bidram, F. L. Lewis; A. Davoudi,
Distributed Control Systems for Small-Scale Power Networks: Using Multiagent Cooperative Control Theory,
{\it IEEE Control Systems},
vol. 34, no. 6,
pp. 56-77,
2014.
\bibitem{CH}
D. S. Callaway, I. A. Hiskens,
{\it Achieving controllability of electric loads},
Proceedings of the IEEE, vol. 99, issue 1, 2011.
\bibitem{LH}
N. Lu, D. J. Hammerstrom,
{\it Design considerations for frequency responsive
grid friendly\texttrademark appliances},
Proceedings of IEEE PES Transmission and Distribution Conference and Exhibition, Dallas, TX, USA, 2006.
\bibitem{SIF}
J. A. Short, D. G. Infield, L. L. Freris,
{\it Stabilization of grid frequency through dynamic demand control},
IEEE Transactions on Power Systems
vol. 22, issue 3, 2007.
\bibitem{Ol1}
Khamisov O.O.,
Direct Disturbance Based Decentralized Frequency Control for Power Systems,
{\it Proceedings of 56th IEEE Conference on Decision and Control},
2017.
\bibitem{WWW}
A. J. Wood and B. F. Wollenberg,
{\it Power Generation, Operation, and Control}, 2nd ed.
NJ, US: John Wiley \& Sons, Inc., 1996.
\bibitem{KN}
P. Kundur,
{\it Power System Stability And Control},
McGraw-Hill, 1994.
\bibitem{BV}
A. R. Bergen and V. Vittal,
{\it Power Systems Analysis}, 2nd ed. Prentice
Hall, 2000.
\bibitem{RKC}
F. R. K. Chung,
{\it Spectral Graph Theory},
American Mathematical Society, 1997.
\bibitem{BRB}
R. B. Bapat,
{\it Graphs and Matrices}, Springer-Verlag London,
2014.
\bibitem{MBB}
J. Machowski, J. Bialek, and J. Bumby,
{\it Power System Dynamics: Stability and Control}, 2nd ed.
NJ, US: John Wiley \& Sons, Inc.,
2008.
\bibitem{ZL}
C. Zhao, S. Low,
Optimal decentralized primary frequency control in power networks,
{\it 53rd IEEE Conference on Decision and Control},
pp. 2467-2473,
2014.
\bibitem{ZML}
C. Zhao, E. Mallada, S. Low,
Distributed generator and load-side secondary frequency control in power networks,
{\it 2015 49th Annual Conference on Information Sciences and Systems (CISS)}.
pp. 1-6,
2015.
\bibitem{B}
S. Boyd, L. Vandenberghe,
{\it Convex optimization},
Cambridge University Press,
2004.
\bibitem{D}
B. P. Demidovich,
{\it Lectures on the Mathematical Stability Theory},
Nauka, Moscow,
1967 (in Russian).
\bibitem{CR}
J. Chow and G. Rogers,
Power system toolbox,
{\it Cherry Tree Scientific Software},
2000.
\end{thebibliography}
\end{document}